\documentclass[prodmode,acmec]{ec-acmsmall} 

\usepackage[ruled]{algorithm2e}

\newtheorem{algo4.9}{Algorithm 4.9}
\SetAlFnt{\small}
\SetAlCapFnt{\small}
\SetAlCapNameFnt{\small}
\SetAlCapHSkip{0pt}
\IncMargin{-\parindent}
\newcommand{\mtr}{\textsc{mtr}}

\usepackage[numbers]{natbib}
\makeatletter
\newcommand{\Rmnum}[1]{\expandafter\@slowromancap\romannumeral #1@}
\begin{document}

\markboth{}{}

\title{ Computational issues in time-inconsistent planning}
\author{PINGZHONG TANG
\affil{IIIS, Tsinghua University}
YIFENG TENG
\affil{IIIS, Tsinghua University}
ZIHE WANG
\affil{IIIS, Tsinghua University}
SHENKE XIAO
\affil{IIIS, Tsinghua University}
YICHONG XU
\affil{IIIS, Tsinghua University}
}

\begin{abstract}
Time-inconsistency refers to a paradox in decision making where agents exhibit inconsistent behaviors over time. Examples are {\em procrastination} where agents tends to postpone easy tasks, and {\em abandonments} where agents start a plan and quit in the middle. These behaviors are undesirable in that agents make clearly suboptimal decisions over optimal ones. To capture such behaviors and more importantly, to quantify inefficiency caused by such behaviors, \cite{kleinberg2014time} propose a graph model which is essentially the same as the standard planning model except for the cost structure. Using this model, they initiate the study of several interesting computation problems: 1) cost ratio: the worst ratio between the actual cost of the agent and the optimal cost, over all the graph instances; 2) motivating subgraph: how to motivate the agent to reach the goal by deleting nodes and edges; 3) Intermediate rewards: how to motivate agents to reach the goal by placing intermediate rewards. Kleinberg and Oren give partial answers to these questions, but the main problems are still open. In fact, they raise these problems as open problems in their original paper.

In this paper, we give answers to all three open problems in \cite{kleinberg2014time}. First, we show a tight upper bound of cost ratio for graphs without Akerlof's structure, thus confirm the conjecture by Kleinberg and Oren that Akerlof's structure is indeed the worst case for cost ratio. Second, we prove that finding a motivating subgraph is NP-hard, showing that it is generally inefficient to motivate agents by deleting nodes and edges in the graph. Last but not least, we show that computing a strategy to place minimum amount of total reward is also NP-hard.

\end{abstract}







\maketitle

\section{Introduction}

In behavioral economics, an important theme has been to understand individual behaviors that are inconsistent over time. There are at least two types of inconsistencies investigated in the literature. The first type is {\em Procrastination}~\cite{akerlof1991procrastination,o1999doing,kleinberg2014time}: agents tend to postpone costly actions even though such delay may incur further cost.  The second type is {\em abandonment}~\cite{abandonment2008}: agents plan for a multi-phase task (usually with rewards in the end), spend efforts in the initial phases and decide to quit in the middle.

Both types of behaviors have been widely observed in reality. Akerlof~\shortcite{akerlof1991procrastination} describes a story of procrastination (restated in~\cite{kleinberg2014time}) where an agent must ship a package within the next a few days, incurs an immediate cost for shipping the package or some additional daily cost for not shipping the package. Clearly, the {\em optimal} strategy for the agent is to ship the package right away, avoiding any additional daily cost. However, as the story goes, the agent chooses to procrastinate and to send the package in one of the last few days. Similar examples abound, ranging from golf club members that never play golf (abandonment) to investors that rent an apartment for years before making a purchase (procrastination).

The interpretation to all these phenomena lies in that agents value current cost more than the cost in the future. Researchers in the literature have developed various models to capture this observation and to interpret the inconsistencies~\cite{strotz1955myopia,akerlof1991procrastination,laibson1997golden,frederick2002time}. We refer the readers to \cite{kleinberg2014time} and the references therein for a detailed description of this line of work. In what follows, we describe the model in~\cite{kleinberg2014time}, coined the {\em time-inconsistent planning} model, based on which all the analyses of this paper are built.

\subsection{The time-inconsistent planning model}
Roughly put, the time-inconsistent planning model is no different from the standard planning model~\cite{pollak1968consistent,Russell2003}, except for a slight twist on the cost structure. The standard planning model is a directed graph (aka. task graph) where each node in the graph represents a state, each directed edge denotes an action that transits one state to another and each action incurs a certain cost, marked as the weight on the edge. The planner's goal is to find a shortest (min-cost) path between the initial state and goal state. The time-inconsistent planning model modifies the model above by redefining the cost of a path: instead of summing the costs of all edges on that path: $\sum_{i=1} c(e_i)$ where $c(e_i)$ is the cost of the $i$-th edge on the path, a time-inconsistent agent applies a multiplicative factor $0<\beta<1$ to the costs of all the edges except for the first edge in the path $c(e_1)+\beta\sum_{i=2} c(e_i)$. The interpretation is that the time-inconsistent agent evaluates actions at the current state at its true cost, while discounts the costs (rewards) of all future actions by $\beta$. The time-inconsistent planning model is defined as a time-inconsistent agent who looks for a discounted shortest path at any state in the task graph. The cost model above can also be considered as a special case of the {\em quasi-hyperbolic discounting} model~\cite{laibson1997golden}.

Simple as it appears to be, the model is powerful enough to capture a range of interesting time-inconsistent behaviors. In particular, Kleinberg and Oren~\cite{kleinberg2014time} show that:
\begin{itemize}
\item It can easily capture Akerlof's example of procrastination by Figure~\ref{fig:one}.
\item Time-inconsistency agents sometimes follows a sub-optimal path (see also Figure~\ref{fig:one}).
\item The model can be easily extended to a model to include reward, simply by placing some reward at the goal node. This extension can be further used to capture the phenomenon of abandonment: an agent may find it desirable to follow the optimal discounted path at initial nodes but then find it not beneficial when evaluating at some middle node.
\item The model can be used to model the interesting fact of {\em choice reduction}: agents can be better motivated to reach the goal by deleting certain middle nodes and edges from the task graph.
\end{itemize}

\subsection{Results, open problems by ~\cite{kleinberg2014time} and our contributions}

Perhaps more importantly (from the perspective of EC), the time-inconsistent model facilitates analyses of the following important computation problems:
\begin{enumerate}
\item Cost ratio. Cost ratio is defined as the ratio between the cost of the path found by the agent and that of the min-cost path. \cite{kleinberg2014time} gives a characterization of cost ratio in terms of {\em graph minors\footnote{To be formally defined immediately}}: roughly, any graph with a sufficiently high (at least $\lambda^n$ for some $\lambda >1$ and $n$) cost ratio must contain Figure~\ref{fig:one} (denoted as $\mathcal{F}_k$) as a graph minor and $k$ is at least a constant fraction of $n$. In light of this characterization, an important open problem raised by Kleinberg and Oren is: when the graph does not contain a $\mathcal{F}_k$-minor, how bad can the cost ratio be? This question is particularly important since it concerns whether $\mathcal{F}_k$ is indeed the worst case instance for cost ratio.

We solve this problem by proving that, for any graph that does not contain a $\mathcal{F}_k$-minor, the cost ratio can be at most $\beta^{2-k}$, where $\beta$ is the discount factor. Therefore, we confirm that $\mathcal{F}_k$ is indeed the worst case instance for cost ratio and the bound proved by Kleinberg and Oren is tight.

\item Minimal Motivating Subgraphs. A motivating subgraph is a subgraph of the original task graph and a time-inconsistent agent can reach the goal in this subgraph. A motivating subgraph is minimal if none of its proper subgraphs is motivating. Clearly, motivating subgraph is closely related to the previously mentioned economic problem of choice reduction. Kleinberg and Oren prove a relatively complex property that says the minimal motivating subgraphs are necessarily {\em sparse}. Here, we ask a natural complexity question: what is the computational complexity of finding motivating subgraph? This is also the second open problem raised by Kleinberg and Oren. We prove that this problem is NP-hard. More generally, we show that finding any motivating subgraph (e.g., maximal) is NP-hard.

\item Cost of placing intermediate rewards. Instead of motivating agents to reach their goal by choice reduction (i.e., via motivating subgraph), an alternative way that has been seen in the literature is to place rewards on intermediate nodes. A natural question (the third open problem by Kleinberg and Oren) is: what is the minimum total reward needed to motivate an agent to reach its goal? We prove that this problem is also NP-hard.
\end{enumerate}

In short, we give answers to all three open questions in \cite{kleinberg2014time}.

\section{Formal description of the model}

As defined in~\cite{kleinberg2014time}, the task graph is an acyclic directed graph $G$ with a start node $s$ and a target (goal) node $t$, where each edge $(u,v)$ has a non-negative cost $c(u,v)$. For any pair of nodes $(u,v)$, we denote by $d(u,v)$ the minimum total cost from $u$ to $v$,
\[d(u,v)=\min_{P\in\mathcal{P}(u,v)}\sum_{e\in P}c(e),\]
where $\mathcal{P}(u,v)$ is the set of all possible paths from $u$ to $v$. For simplicity, let $d(v)=d(v,t)$ for any node $v$. Denote the discount parameter by $\beta\in[0,1]$. An agent starts at $s$ and travels towards $t$. In each step, the agent at node $u$ chooses an out-neighbor $v$ that minimizes $c(u,v)+\beta d(v)$ (if more than one node minimizes this value, the agent chooses one arbitrarily). Clearly, the agent cares less about the future for smaller $\beta$. Suppose $P$ is the $s$-$t$ path the agent chooses, the \emph{cost ratio} is defined as
$\sum_{e\in P}c(e)/d(s),$
i.e., the ratio of actual cost to the optimal cost.

To state the first open problem by Kleinberg and Oren, as well as our answer, we need the following definitions, also from ~\cite{kleinberg2014time}.

Given two undirected graphs $H$ and $K$, we say that \emph{$H$ contains a $K$-minor} if we can map each node $\kappa$ of $K$ to a connected subgraph $S_{\kappa}$ in $H$, with the properties that (i) $S_{\kappa}$ and $S_{\kappa'}$ are disjoint for every two nodes $\kappa$, $\kappa'$ of $K$, and (ii) if $(\kappa,\kappa')$ is an edge of $K$, then in $H$ there is some edge connecting a node in $S_{\kappa}$ with a node in $S_{\kappa'}$.

Moreover, let $\sigma(G)$ denote the {\em skeleton} of $G$, the undirected graph obtained by removing the directions on the edges of $G$. Let $\mathcal{F}_k$ denote the graph with nodes $v_1,v_2,\ldots,v_k,$ and $w$, and edges ($v_i,v_{i+1}$) for $i=1,\ldots,k-1$, and $(v_i,w)$ for $i=1,\ldots,k$. Figure~\ref{fig:one} depicts $\mathcal{F}_k$. $\mathcal{F}_k$ is a special structure in this setting. For one, it vividly illustrates the Akerlof story. Furthermore, \cite{kleinberg2014time} has proved the following theorem:
\begin{theorem}
~\cite{kleinberg2014time}
For every $\lambda>1$ there exist $n_0>0$ and $\varepsilon>0$ such that if $n\geq n_0$ and cost ratio $r>\lambda^n$, then $\sigma(G)$ contains an $\mathcal{F}_k$-minor for some $k\geq\varepsilon n$.
\label{cost}
\end{theorem}

\begin{figure}
\centerline{\includegraphics[scale=0.4]{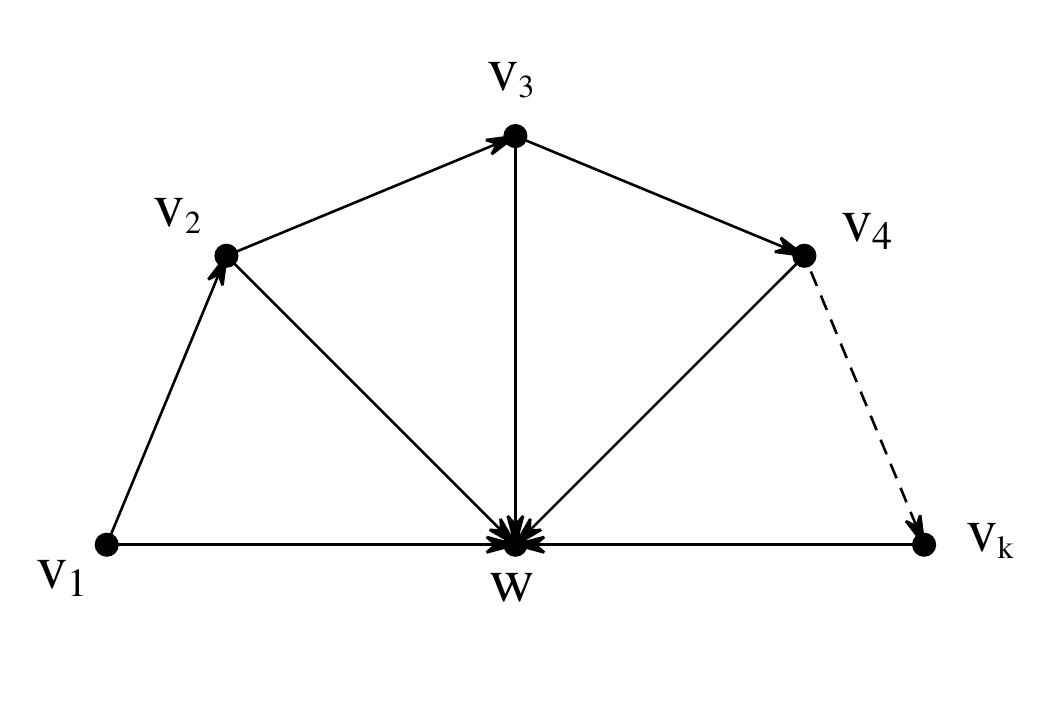}}
\caption{$\mathcal{F}_k$}
\label{fig:one}
\end{figure}

\section{Maximum Cost Ratio}

In light of the result above, Kleinberg and Oren propose the following open question: what is the maximum cost ratio (if exists) if $\sigma(G)$ does not contain an $\mathcal{F}_k$-minor. In other words, the problem asks, without $\mathcal{F}_k$,  how much waste can be resulted from time-inconsistency. The question is extremely important since it is closely related to whether $\mathcal{F}_k$ is the worst case instance for cost ratio.

Note that an edge is exactly $\mathcal{F}_1$. Assume $k>1$, we have the following theorem,
\begin{theorem}
For any $k>1$, if $\sigma(G)$ does not contain an $\mathcal{F}_k$-minor, the cost ratio is at most $\beta^{2-k}$. This bound is tight and can be achieved by $\mathcal{F}_{k-1}$.
\label{cost}
\end{theorem}


\subsection{Proof Sketch}

To analyze the cost ratio of any graph, our first observation is to focus on the set of \emph{shortcut nodes}. Roughly, a shortcut node is one where the agent's min-cost choice is different from his actual choice. Clearly, if there were no such nodes, i.e., the agent's actual path and min-cost path coincide, we would end up in the ideal case where the cost ratio is 1.
For each shortcut node, we obtain an inequality that states the discounted cost of the actual path is less than or equal to that of the min-cost path. The intuition here is that each appearance of such a shortcut node contribute a factor of $\beta$ to the cost ratio, and $k$ appearances (to be rigourously defined in the main proof) would lead to the worst case of $\beta^k$ and this only happens if $\mathcal{F}_k$-minor exists.

To formally prove this statement, we need to carefully expand the cost formula $d(s)$ as a linear combination of costs on edges $c(e)$'s. This is complicated, again, by the existence of shortcut nodes, since the recursive formula that defines $d(u_i)$ (where $u_i$ is some shortcut node) introduces two new terms $d(u_i')$, the cost from the next node on the actual path and $d(w_i)$, the cost where the current min-cost path merges with the actual path (See Figure \ref{fig:rel}). Our strategy is to fix $d(w_i)$ and carefully expand $d(u_i')$.

A key step of our proof is that three different cases of $w_i$ are considered and different relaxations are given for each case (see the definition of $t_i$ below). This is also why Kleinberg-Oren paper fails to get the tight bound.

Continue the expansion of $d(u_i)$ using the rules above until the right hand side contains $c(e)$'s only, i.e., representing $d(u_i)$ as a linear combination of $c(e)$'s. We obtain the final bound by bounding the coefficients of the linear combination.

\subsection{Formal Proof}
Now we give a formal proof to Theorem \ref{cost}. The proofs of all the lemmas will be shown in appendix.
\begin{proof}
Let $P$ be the path that the agent actually travels through. The main idea of the proof is to obtain an inequality with the form $d(s)\geq \sum_{e\in P}\alpha(e)c(e)$ where $\alpha$'s are positive coefficients, so that we can use the minimal coefficient to bound $d(s)/\sum_{e\in P}c(e)$.

Firstly we introduce some notations. For any node $u$ on $P$, denote by $u'$ the node immediately after $u$ in $P$; for any pair of nodes $(u,v)$ in $P$, denote by $c(u,v)$ the total cost of edges between $u$ and $v$ ion $P$. A node $u$ in $P$ is defined to be a \emph{shortcut node} if the second node on the min-cost path from $u$ to $t$ is not $u'$ (note that the first node is $u$). In other words, a shortcut node is one where the agent's min-cost choice is different from his actual choice on $P$.

If $u$ on $P$ is not a shortcut node, we have $d(u)=c(u,u')+d(u')$, then for the ideal case when there is no shortcut node, we would have $d(s)=c(s,s')+d(s')=c(s,s')+c(s',s'')+d(s'')=\cdots=\sum_{e\in P} c(e)$, resulting in a cost ratio of $1$.

 Now, suppose there are $n$ shortcut nodes: $u_1,u_2,\ldots,u_n$ by the order of appearance on $P$. For $i=1,2,\ldots,n$, denote by $P_i$ the min-cost path from $u_i$ to $t$; denote by $w_i$ the second crossing point (note that the first node is $u_i$) of $P$ and $P_i$ (if there are more than one min-cost path, arbitrarily choose one); denote by $v_i$ the first node after $u_i$ on $P_i$. Figure \ref{fig:rel} describes the notations above.

\begin{figure}
\centerline{\includegraphics[scale=0.6]{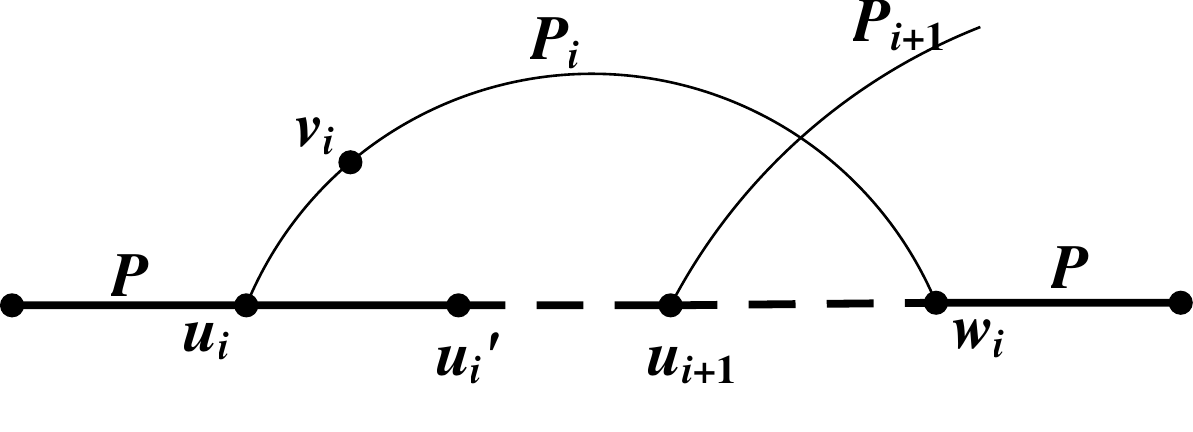}}
\caption{Relationship between nodes.}
\label{fig:rel}
\end{figure}

By the definition of time-inconsistency, the agent at $u_i$ chooses $P$ over $P_i$, we have
\begin{eqnarray*}
d(u_i,w_i)+\beta d(w_i)\geq c(u_i,v_i)+\beta d(v_i,w_i)+\beta d(w_i)\geq c(u_i,u_i')+\beta d(u_i').
\end{eqnarray*}
\indent Add $(1-\beta)d(w_i)$ to both sides of the inequality and use $d(u_i')=c(u_i',u_{i+1})+d(u_{i+1})$, we have
\begin{eqnarray}
d(u_i)&\geq & c(u_i,u_i')+\beta d(u_i')+ (1-\beta)d(w_i)\nonumber\\
&=&c(u_i,u_i')+\beta c(u_i',u_{i+1})+\beta d(u_{i+1})+ (1-\beta)d(w_i). \label{xsk:core inequality}
\end{eqnarray}
\indent Formula (\ref{xsk:core inequality}) is an important inequality for us to obtain our final inequality. We will expand the righthand side of (\ref{xsk:core inequality}) iteratively. For concreteness, we list below the first few steps of the expansion:
\begin{eqnarray*}
d(s)&=& c(s,u_1)+d(u_1) \\
&\geq & c(s,u_1)+c(u_1,u_1')+\beta d(u_1')+(1-\beta)d(w_1)\\
&= & c(s,u_1)+c(u_1,u_1')+\beta c(u_1',u_2)+\beta d(u_2)+(1-\beta)d(w_1) \\
&\geq & c(s,u_1)+c(u_1,u_1')+\beta c(u_1',u_2)+\beta c(u_2,u_2')+\beta^2 d(u_2')\\
&&+\:\beta(1-\beta)d(w_2)+(1-\beta)d(w_1) \\
&\geq &\cdots.
\end{eqnarray*}
We now claim that
\begin{eqnarray}
d(s)\geq\sum_{j=1}^i \left(a_jc(u_{j-1}',u_j)+b_jc(u_j,u_j')\right)+a_{i+1}d(u_i')+\sum_{j\in S_{i+1}}(1-\beta)b_jd\left(w_j\right), \label{xsk:expand form}
\end{eqnarray}
where $\{a_i,i=1,2,\ldots\},\{b_i,i=1,2,\ldots\}$ are coefficients to be determined, and $S_i$ is the set of $i$'s such that $w_i$ lies after $u_{i-1}'$ (if $u_{i-1}'=u_i$, then $u_{i-1}'$ is included). Here $u_0'=s$ and $u_{n+1}=t$. We will prove this claim by induction. Before that, for convenience, for $i=1,2,\ldots,n$, define $t_i$ as follows: if there exists $j$ such that $w_i=u_j$, then $t_i=j$; otherwise if there exists $j$ such that $w_i=u_j'$, then $t_i=j+0.5$; otherwise $t_i$ is the smallest index such that $u_{t_i}$ lies after $w_i$ in $P$ (if no such index, $t_i=n+1$). Since $w_i$ must lie after $u_i'$ in $P$, we can obtain a trivial property that $t_i\geq i+1$. Now $S_i$ can be represented by $\{j|1\leq j<i,t_j\geq i\}$.

For $i=1$, this claim holds trivially.

For inductive cases: if (\ref{xsk:expand form}) holds for $i-1$, then
\begin{eqnarray}
d(s)&\geq &\sum_{j=1}^{i-1} \left(a_jc(u_{j-1}',u_j)+b_jc(u_j,u_j')\right)+a_{i}d(u_{i-1}')+\sum_{j\in S_{i}}(1-\beta)b_jd\left(w_j\right) \nonumber\\
&=&\sum_{j=1}^{i-1} \left(a_jc(u_{j-1}',u_j)+b_jc(u_j,u_j')\right)+a_{i}\left(c(u_{i-1}',u_i)+d(u_i)\right)+\sum_{j:t_j=i}(1-\beta)b_jd\left(w_j\right) \nonumber\\
&&+\:\sum_{j:i<t_j<i+1}(1-\beta)b_jd\left(w_j\right)+\sum_{j\in S_{i}:t_j\geq i+1}(1-\beta)b_jd\left(w_j\right)\nonumber\\
&\geq &\sum_{j=1}^{i-1} \left(a_jc(u_{j-1}',u_j)+b_jc(u_j,u_j')\right)+a_{i}c(u_{i-1}',u_i)+\left(a_i+\sum_{j:t_j=i}(1-\beta)b_j\right)d(u_i)\nonumber\\
&&+\:\sum_{j:i<t_j<i+1}(1-\beta)b_jd\left(u_i'\right)+\sum_{j\in S_{i}:t_j\geq i+1}(1-\beta)b_jd\left(w_j\right) \label{xsk:tricky bound}
\\
&\geq &\sum_{j=1}^{i-1} \left(a_jc(u_{j-1}',u_j)+b_jc(u_j,u_j')\right)+a_{i}c(u_{i-1}',u_i)\nonumber\\
&&+\:\left(a_i+\sum_{j:t_j=i}(1-\beta)b_j\right)\left(c(u_i,u_i')+\beta d(u_i')+ (1-\beta)d(w_i)\right)\nonumber\\
&&+\:\sum_{j:i<t_j<i+1}(1-\beta)b_jd\left(u_i'\right)+\sum_{j\in S_{i}:t_j\geq i+1}(1-\beta)b_jd\left(w_j\right)\nonumber\\
&\geq &\sum_{j=1}^{i-1} \left(a_jc(u_{j-1}',u_j)+b_jc(u_j,u_j')\right)+a_{i}c(u_{i-1}',u_i)+\left(a_i+\sum_{j:t_j=i}(1-\beta)b_j\right)c(u_i,u_i')\nonumber\\
&&+\:\left(\beta \left(a_i+\sum_{j:t_j=i}(1-\beta)b_j\right)+\sum_{j:i<t_j<i+1}(1-\beta)b_j\right)d(u_i')\nonumber\\
&&+\:(1-\beta)\left(a_i+\sum_{j:t_j=i}(1-\beta)b_j\right)d\left(w_i\right)+\sum_{j\in S_{i}:t_j\geq i+1}(1-\beta)b_jd\left(w_j\right). \label{xsk:last step}
\end{eqnarray}
In (\ref{xsk:tricky bound}) we use the property that if $t_j=i$, i.e. $w_j$ lies between $u_{i-1}'$ and $u_i$ (both are included, since $w_j=u_{i-1}'$ in the case that $u_{i-1}'=u_i$) on $P$, then $d(w_j)=c(w_j,u_i)+d(u_i)\geq d(u_i)$. Now if we set
\begin{eqnarray}
a_i=\beta b_{i-1}+\sum_{j:i-1<t_j<i}(1-\beta)b_j,b_i=a_i+\sum_{j:t_j=i} (1-\beta)b_j,\label{xsk:relation between a and b}
\end{eqnarray}
Formula (\ref{xsk:last step}) is almost the same as (\ref{xsk:expand form}) except for the last term. In fact, the following lemma shows (\ref{xsk:last step}) and (\ref{xsk:expand form}) are exactly the same.

\begin{lemma}
For any sequence $\{x_i,i=1,2,\ldots,\}$ and for $m=2,3,\ldots,n$, we have
\[\sum_{j\in S_m}x_j=x_{m-1}+\sum_{j\in S_{m-1}:t_j\geq m}x_j.\] \label{xsk:mathematical lemma}
\end{lemma}

By this lemma, we see that (\ref{xsk:last step}) and (\ref{xsk:expand form}) are of the same form, thus we have proved claim (\ref{xsk:expand form}) by induction. Furthermore, $a_i$ and $b_i$ are determined by (\ref{xsk:relation between a and b}) with $a_1=1$.

Now set $i=n$ in claim (\ref{xsk:expand form}), we have
\begin{eqnarray}
d(s)\geq\sum_{j=1}^n \left(a_jc\left(u_{j-1}',u_j\right)+b_jc\left(u_j,u_j'\right)\right)+a_{n+1}c(u_n',t).\label{xsk:result with a and b}
\end{eqnarray}
Then all we need to do is to find lower bounds for $a_i$ and $b_i$ respectively. The bounds can be obtained by the following lemma.

\begin{lemma}
$b_i\geq a_i\geq \beta^{|S_i|}$. \label{xsk:bound lemma}
\end{lemma}

Finally we show $|S_i|\leq k-2$ for $i=1,2,\ldots,n+1$. For contradiction, let us assume $|S_i|\geq k-1$ for some $i$. Choose $k-1$ elements from $S_i$, say $j_1,j_2,\ldots,j_{k-1}$ where $j_1<j_2<\cdots<j_{k-1}<i$. For $l=1,2,\ldots,k-1$, since $t_{j_l}\geq i$, we have $w_{j_l}$ must lie after $u_{i-1}$, or $u_{j_{k-1}}$ on $P$, and also $w_{j_l}\neq u_{j_{k-1}}'$ (otherwise $t_{j_l}=j_{k-1}+0.5<i$), $w_{j_l}$ must lie after $u_{j_{k-1}}'$ on $P$. Now consider $u_{j_1},u_{j_2},\ldots,u_{j_{k-1}},u_{j_{k-1}}'$. Observe that the nodes lie on path $P$ in order. For $l=1,2,\ldots,k-1$, $P_{j_l}$ starts at $u_{j_l}$, and the next crossing point of $P_{j_l}$ and $P$ lies after $u_{j_{k-1}}'$ in $P$.

The following lemma states that if a graph has such a structure, it must contain an $\mathcal{F}_k$-minor.
\begin{lemma}
Let $P$ be a path of $G$, and $u_1,u_2,\ldots,u_k$ are nodes on $P$ in order of appearance. If for $i=1,2,\ldots,k-1$, there exists a path $P_i$ such that (i) it starts at $u_i$, and (ii) the second crossing point in $P_i$ with $P$ (say $u_i'$) exists and lies after $u_k$ on $P$, then $\sigma(G)$ contains an $\mathcal{F}_k$-minor. \label{xsk:graph lemma}
\end{lemma}

By this lemma we conclude that $\sigma(G)$ contains an $\mathcal{F}_k$-minor, a contradiction. Thus $|S_i|\leq k-2$.\\
\indent Hence by (\ref{xsk:result with a and b}) and Lemma \ref{xsk:bound lemma} we have
\begin{eqnarray*}
d(s)\geq\sum_{j=1}^n \left(\beta^{|S_j|}c\left(u_{j-1}',u_j\right)+\beta^{|S_j|}c\left(u_j,u_j'\right)\right)+\beta^{|S_{n+1}|}c\left(u_n',t\right)\geq \beta^{k-2}c(s,t),
\end{eqnarray*}
which implies
$c(s,t)/d(s)\leq \beta^{2-k}$.
\end{proof}

\subsection{Tightness of the Bound}
So far we have obtained an upper bound for $r$, now we provide an example to show that the bound is achievable. We simply use the example mentioned in \cite{kleinberg2014time} to show a graph with exponential cost ratio, i.e. the graph obtained from $\mathcal{F}_{k-1}$ by adding the corresponding weights (see Figure \ref{akerlof's example}).
\begin{figure}
\centerline{\includegraphics[scale=0.6]{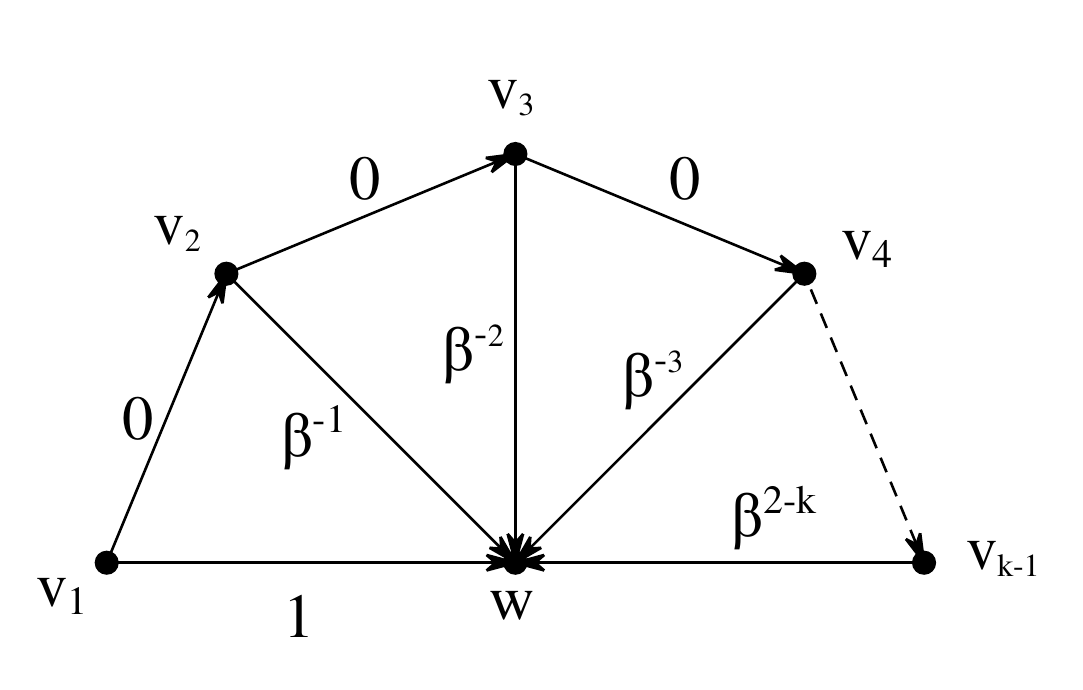}}
\caption{Akerlof's example}
\label{akerlof's example}
\end{figure}

By the analysis in \cite{kleinberg2014time}, the cost ratio of the graph is exactly $\beta^{2-k}$, which proves the tightness of our upper bound.

\section{Hardness of Finding minimal motivating subgraphs}

As mentioned, the basic model introduced in Section 2 can be easily extended to capture abandonment, by placing a reward at the target node. Formally, if the reward is $r$ and the agent is in node $u$, if $\min_v c(u,v)+\beta d(v)>\beta r$, i.e., the discounted cost is less than the discounted reward, the agent abandons the plan.

A natural question in this extended model is whether a time-inconsistent agent can reach the target, and if not, can we delete some nodes and edges to help it reach the goal. The first question is easy to check. To formally investigate the second question, define {\em motivating subgraph} as a subgraph of the original task graph such that the agent can reach the target in the subgraph. A motivating subgraph is minimal if none of its proper subgraph is motivating. We are interested in the following computational question concerning (minimal) motivating subgraph (also the second the open question listed in \cite{kleinberg2014time}): {\em is there a polynomial time algorithm that finds a (minimal) motivating subgraph?}

In what follows, we answer this question negatively (unless NP=P) with the following theorem.

\begin{definition}
{\bf Problem \textsc{ms}}: for an acyclic graph \(G\) with \(n\) nodes, given reward \(r\) on target node and bias factor \(\beta\), find a motivating subgraph of \(G\).
\end{definition}

\begin{theorem}\label{thm:4.2}
\textbf{Problem \textsc{ms}} is NP-hard.
\end{theorem}

Before proving Theorem \ref{thm:4.2}, consider an easier complexity problem related to minimal motivating subgraphs.

\subsection{Hardness of finding minimal motivating subgraph}
In this section, we show that finding a minimal motivating subgraph is hard.
\begin{definition}
\textbf{Problem \textsc{mms}}: {\it for an acyclic graph \(G\) with \(n\) nodes, given reward \(r\) on target node and present bias \(\beta\), find a minimal motivating subgraph of \(G\)}.
\end{definition}

We have the following theorem.

\begin{theorem}\label{thm:4.4}
\textbf{Problem \textsc{mms}} is NP-hard.
\end{theorem}
\begin{proof}
We show that finding a valid assignment to a 3-CNF can be polynomial-time reduced to an instance of \textsc{mms}. Consider a 3-CNF with \(n\) variables \(x_1,x_2,...,x_n\) and \(m\) clauses \(C_1,C_2,...,C_m\). For fixed \(\beta\), let \(f=1-\frac{1}{2}\beta-\frac{1}{2}\beta^2\), \(z=2+\frac{3}{2}\beta+\frac{1+\beta}{1-\beta}\), \(\ell=\lceil\frac{z}{f}\rceil\). Construct a weighted acyclic graph \(G\) as follows:
\begin{enumerate}
\item Each clause \(C_i\) corresponds to one node \(u_i\), \(1\leq i\leq m\);

\item Each variable \(x_i\) corresponds to two nodes \(v_i\) and \(v'_{i}\), \(1\leq i\leq n\);

\item The other nodes are start node \(s\), target node \(t\) and interior points \(w_1,w_2,\cdots,w_{\ell}\) and \(w\);

 \item For any clause \(C_i=y_{i,1}\vee y_{i,2}\vee y_{i,3}\): for \(j=1,2,3\), if \(y_{i,j}\) is \(x_k\), then there is an edge \((u_i, v_j)\) with weight 2 and we call such an edge an ``expensive edge''; if \(y_{i,j}\) is \(\neg x_k\), then there is an edge \((u_i, v_{j})\) with weight \(1+\beta\) and we call such an edge a ``cheap edge''.

\item For any \(1\leq i\leq n\): there is an edge \((v_i, v'_{i})\) with weight \(1-\beta\) and an edge \((v'_i,w)\) with weight 0. The two edges form an ``expensive path'' from \(v_i\) to \(w\). Also, there is an edge \((v_i, w)\) with weight 0, and we call such an edge a ``cheap path'' from \(v_i\) to \(w\).

\item \((s, u_1)\), \((u_m, w_1)\) are edges with weight \(f\). For each \(1\leq i\leq m-1\), \((u_i, u_{i+1})\) forms an edge with weight \(f\). For each \(1\leq j\leq \ell-1\), \((w_i, w_{i+1})\) forms an edge with weight \(f\). Edge \((w_{\ell},t)\) has weight \(f+z-f\ell\). These edges form a path from \(s\) to \(t\), and we call the path the ``bus'' of the graph. Finally there is an edge \((w,t)\) with weight \(\frac{3}{2}\beta+\frac{1+\beta}{1-\beta}\).
\end{enumerate}

To visualize the construction, Figure 4 is a graph constructed from 3-CNF \((x_1\vee\neg x_2\vee x_3)\wedge(x_2\vee\neg x_3\vee x_4)\) by following the rules above, and in this example \(\beta\) is set to be 0.9.
\begin{figure}[htbp]
\centering\includegraphics[width=4.5in]{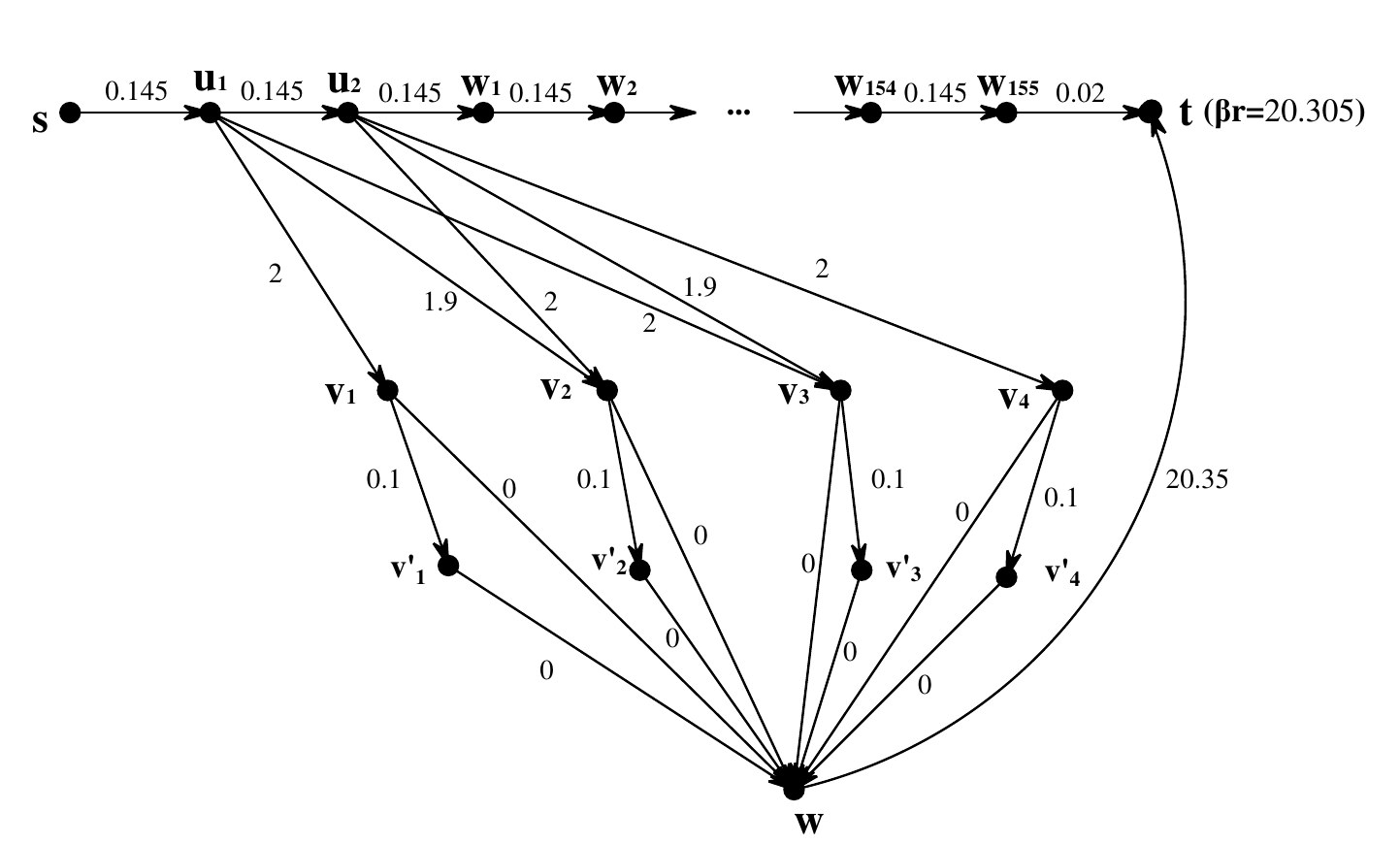}
\caption{Corresponding graph of \((x_1\vee\neg x_2\vee x_3)\wedge(x_2\vee\neg x_3\vee x_4)\) with \(\beta=0.9\).}\label{example0201:1}
\end{figure}

To form an instance of \textsc{mms}, define the reward at \(t\) to be \(r=1+\frac{1}{2}\beta+\beta^{-1}+\frac{2}{1-\beta}\). It can be easily seen that such an instance can be constructed from the 3-CNF in polynomial time. Now we have the following lemmas about the structure of minimal motivating subgraph of \(G\).

\begin{lemma}\label{lemma4.5}
Let \(G'\) be a minimal motivating subgraph of \(G\), then the bus of \(G\) is contained in \(G'\). The only possible route that the agent travels through must be the bus.
\end{lemma}
\begin{proof}
Assume by contradiction that some nodes on the bus are not included in \(G'\). Consider the path from \(s\) to \(t\) traveled by the agent in \(G'\), we may easily observe that \(w\) is on the path. When the agent arrives at \(w\), the evaluated cost to the target will be \(\frac{3}{2}\beta+\frac{1+\beta}{1-\beta}>\beta+\frac{1}{2}\beta^2+\frac{1+\beta}{1-\beta}=\beta r\), thus the agent will not continue, which is a contradiction to the fact that the graph is motivating. Thus the bus of \(G\) is contained in \(G'\). The reasoning also shows that the agent should never deviate from the bus.
\end{proof}
\begin{lemma}\label{lemma4.55}
Let \(G'\) be a minimal motivating subgraph of \(G\). In \(G'\) for any \(1\leq j\leq n\), if \(v_j\in G'\), then in \(G'\) there exists exactly one path from \(v_j\) to \(w\).
\end{lemma}
\begin{proof}
If \(v_j\) starts no path to \(w\), then \(v_j\) cannot reach \(t\) in \(G'\). We can observe that after deleting \(v_j\) the graph is still motivating, thus \(G'\) is not minimal, which is a contradiction. If both paths from \(v_j\) to \(w\) remains in \(G'\), we eliminate the expensive path from \(v_j\) to \(w\) while keeping the graph's property of motivating, which also contradicts the minimality of \(G'\). Thus in \(G'\) \(v_j\) starts a unique path to \(w\).
\end{proof}
\begin{lemma}\label{lemma4.6}
Let \(G'\) be a minimal motivating subgraph of \(G\). In \(G'\) for any \(1\leq i\leq m\), \(u_i\) connects to exactly one node \(v_j\) not on the bus, while path \(u_i\to v_j\) and \(v_j \to w\) cannot be both expensive or both cheap.
\end{lemma}
\begin{proof}
We apply induction on \(i\). Firstly we consider the base case where \(i=m\). Theorem 5.1 in \cite{kleinberg2014time} stated that every node in a minimal motivating subgraph has no more than two outgoing edges. Since \(u_i\) has one out-degree in the bus from Lemma \ref{lemma4.5}, there exists at most one path from \(u_i\) to \(w\). If \(u_i\) connects no edge to nodes outside the bus, at \(u_{i-1}\) the agent will evaluate the cost to \(t\) as \(f+\beta (f+z)>\beta r\), thus the agent will stop at \(u_{i-1}\), contradictory with \(G'\) is motivating. Now we assume \(u_i\) connects to node \(v_j\), from Lemma \ref{lemma4.55} \(v_j\) starts exactly one path to \(w\). If \(u_i\to v_j\) and \(v_j \to w\) are both expensive, then the cost from \(u_i\) to \(w\) will be \(3-\beta\). At \(u_{i-1}\) the agent will evaluate the cost to \(t\) as \(f+\beta(3-\beta+\frac{3}{2}\beta+\frac{1+\beta}{1-\beta})>\beta r\), thus the agent will stop at \(u_{i-1}\), which forms a contradiction. If both paths are cheap, the agent will consider traveling through the bus (with evaluated cost \(f+\beta z\)) more expensive than going out of the bus (with cost \(1+\beta+\beta(\frac{3}{2}\beta+\frac{1+\beta}{1-\beta})\)), which means that the agent cannot reach the goal following Lemma \ref{lemma4.5}.  Thus the statement is true when \(i=m\).

Assume that the lemma is true for \(i=k\), consider the case where \(i=k-1\). If \(u_i\) connects no edge to nodes outside the bus, at \(u_{i-1}\) the agent will evaluate the cost of continuing as at least \(f+\beta (f+2+\frac{3}{2}\beta+\frac{1+\beta}{1-\beta})>\beta r\) by induction. Thus the agent stops at \(u_{i-1}\), contradictory with \(G'\) is motivating. Now we assume \((u_i,v_j)\in G'\), from Lemma \ref{lemma4.55} \(v_j\) starts a unique path to \(w\). The same as the above case \(u_i\to v_j\) and \(v_j \to w\) cannot be both expensive. If both paths are cheap, at \(u_{i-1}\) the agent will consider traveling through the bus (with evaluated cost \(f+\beta (2+\frac{3}{2}\beta+\frac{1+\beta}{1-\beta})\)) more expensive than going out of the bus (with cost \(1+\beta+\beta(\frac{3}{2}\beta+\frac{1+\beta}{1-\beta})\)), which means that the agent deviate from the bus and will not reach the goal following Lemma \ref{lemma4.5}. Therefore the statement is true for \(i=k-1\), thus the lemma is true for all \(1\leq i\leq m\).
\end{proof}

Now we show how a valid assignment to a 3-CNF can be transformed to a corresponding minimal motivating subgraph of \(G\).
\begin{lemma}\label{lemma4.7}
If a 3-CNF is satisfiable, then graph \(G\) constructed from the formula contains a minimal motivating subgraph.
\end{lemma}
\begin{proof}
Construct a subgraph \(G'\) of the original graph \(G\) as follows. For each \(1\leq i\leq n\), if \(x_i\) is assigned true, remove \(v'_i\) and adjacent edges; otherwise remove edge \((v_i,w)\). For \(1\leq i\leq m\), let \(C_i=y_{i_1}\vee y_{i_2}\vee y_{i_3}\), here \(y_{i_j}\) denotes \(x_{i_j}\) or \(\neg x_{i_j}\). Assume WLOG. that \(y_{i_j}\) is assigned true, then preserve only \((u_i,v_{i_j})\) among \((u_i,v_{i_1}),\ (u_i,v_{i_2})\) and \((u_i,v_{i_3})\), eliminate the other two edges from \(G\). We can verify that the remaining graph is motivating as the time-inconsistent agent will always follow the bus.

Finally, we remove all nodes with no in-degree. The minimality of the graph follows from Lemma \ref{lemma4.5} and \ref{lemma4.6}.
\end{proof}

To see how Lemma \ref{lemma4.7} works, again consider formula \((x_1\vee\neg x_2\vee x_3)\wedge(x_2\vee\neg x_3\vee x_4)\) as an example. A valid assignment is \(x_1=\neg x_2=\neg x_3=x_4=1\), and the corresponding minimal motivating subgraph is shown in Figure \ref{fig:5}. The same as before we set \(\beta=0.9\).
\begin{figure}[htbp]
\centering\includegraphics[width=4.2in]{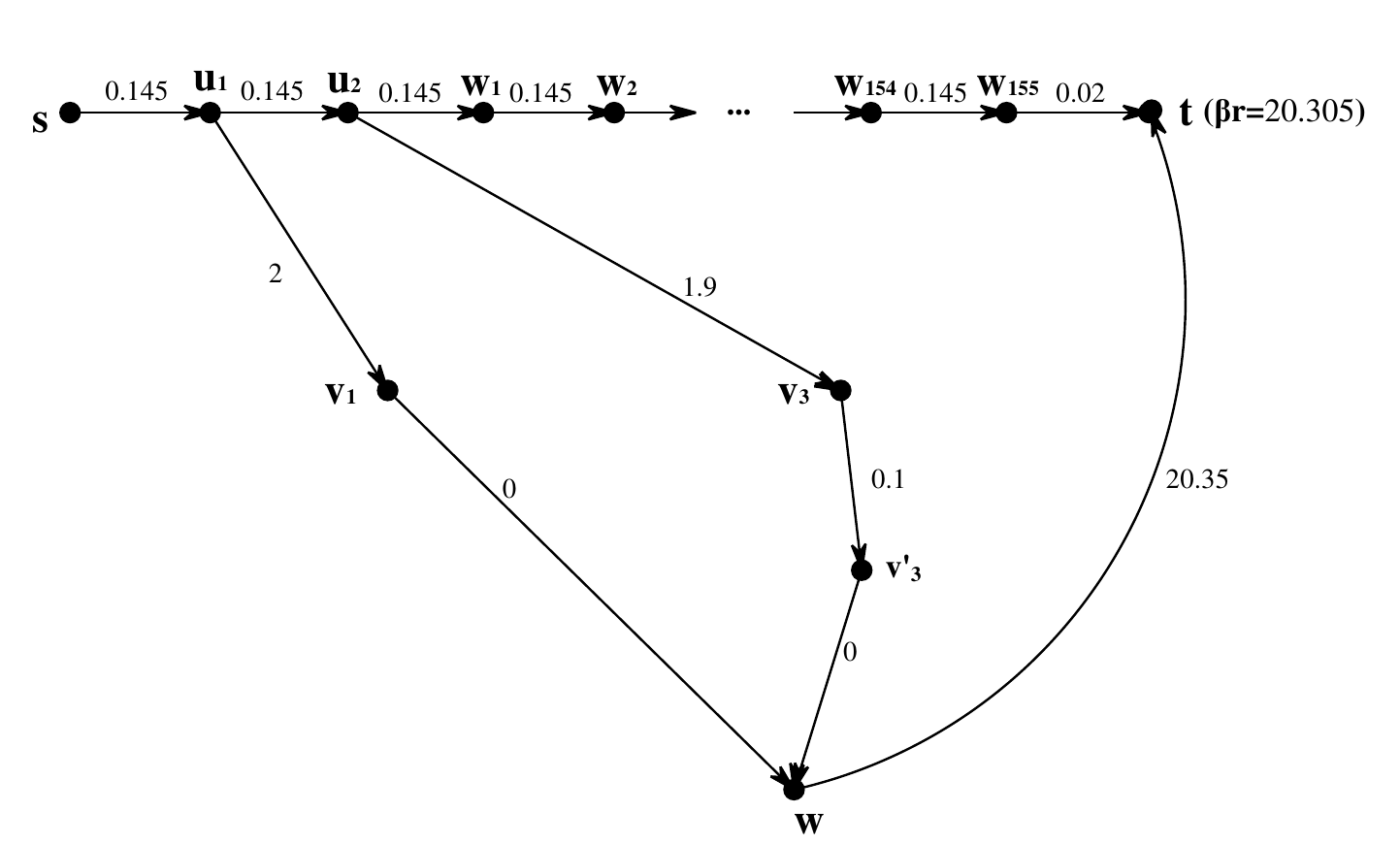}
\caption{Corresponding minimal motivating subgraph of \((x_1\vee\neg x_2\vee x_3)\wedge(x_2\vee\neg x_3\vee x_4)\).}\label{fig:5}
\end{figure}

Finally, we show how a minimal motivating subgraph of \(G\) can be transformed to a valid assignment to the 3-CNF.
\begin{lemma}\label{lemma4.8}
Let \(G'\) be a minimal motivating subgraph of \(G\). For any \(1\leq j\leq n\), if the expensive path from \(v_j\) to \(w\) remains, assign \(x_j=0\), otherwise assign \(x_j=1\), then the assignment is a valid assignment to the original 3-CNF.
\end{lemma}
\begin{proof}
From Lemma \ref{lemma4.6} we know that for any \(1\leq i\leq m\), \(u_i\) connects to exactly one node \(v_j\) outside the bus. If the path from \(v_j\) to \(w\) is cheap, \((u_i,v_j)\) must be expensive, implying that \(x_j\) is in clause \(C_i\), thus the assignment \(x_j=1\) can satisfy \(C_i\). Otherwise, \((u_i,v_j)\) must be cheap, implying that \(\neg x_j\) is in clause \(C_i\), therefore \(C_i\) is satisfied by the assignment letting \(x_j=0\). Thus all clauses are satisfied by the assignment, which means the assignment is a valid assignment to the original 3-CNF.
\end{proof}

From Lemma \ref{lemma4.7} and Lemma \ref{lemma4.8}, we know that if we can find a minimal motivating subgraph for any graph in polynomial time, we may find a polynomial time algorithm finding a valid assignment for any 3-CNF, which finishes the proof of Theorem \ref{thm:4.4}. \(\hfill \Box\)
\end{proof}

\subsection{Main Proof}
Now we are going to prove Theorem \ref{thm:4.2}: it is hard to find any motivating subgraph, not just minimal ones.
\begin{proof}
Assume by contradiction we have a poly-time algorithm \(\mathcal{A}\) which can solve \textsc{ms}. Notice that we can use \(\mathcal{A}\) to check whether a graph has a motivating subgraph. Now we propose the following algorithm that solves \textsc{mms} by calling \(\mathcal{A}\).

For a given graph \(G\), check whether \(G\) contains a motivating subgraph. If not, we reject the input. Repeatedly remove an edge from \(G\) such that the remaining graph still includes a motivating subgraph, until no edge can be removed.

The correctness of the algorithm is straightforward, while the running time of the algorithm is polynomial of the size of the graph. From Theorem \ref{thm:4.4} we know that \textsc{ms} is NP-hard.
\end{proof}

\section{Hardness of motivating agents by placing intermediate reward}

In this section, we consider the question of whether we can place intermediate reward on internal nodes to motivate the agent. Denote by $r(v)$ the reward for node $v$; the agent at node $u$ will continue to move if and only if there exists a path $P$ from $u$ to $t$ such that
\[c'(P)=c(u,v_0)+\beta\sum_{v\in P,v\ne u,t} (c(v,v')-r(v))\leq 0,\]
where $v_0$ denotes the first node after $u$ in $P$, and $v'$ denotes the node after $v$ on $P$. Let $c(t,t')=0$. The agent chooses the path that minimize $c'(P)$. We are interested in finding the configuration that places the smallest total sum of reward that motivate the agent to reach $t$ (which is the third open problem in \cite{kleinberg2014time}).

\subsection{Problem statement}


As suggested by \cite{kleinberg2014time}, multiple versions of this problem are considered. In the first two versions, we restrict attention to positive rewards. In the first version, we might put rewards that are never claimed. We avoid this in the second version: all rewards that we put must be claimed. In the third version, we consider the possibility of placing negative rewards, which can be used to prevent the agent from entering bad paths. In this version, we want to minimize the sum of absolute values of all rewards.


We formally define this problem:
\begin{definition}
\emph{Minimum Total Rewards with bias factor $\beta$ (\mtr$_\beta$)}
Given a weighted acyclic graph $G$, node $s,t$, and a real number $R$, decide whether there exists a reward configuration $r(v),\forall v\in G$ such that the agent is motivated to reach the goal, and that
 $\sum_{v\in G} |r(v)|\leq R.$

There are three versions of the problem, depending on the constraints on the rewards:\\
\mtr$_\beta$ \Rmnum{1}: $r(v)\geq 0,\forall v$.\\
\mtr$_\beta$ \Rmnum{2}: $r(v)\geq 0, \forall v\in P; r(v)=0,\forall v\not\in P$, where $P$ is the path that is actually taken by the agent.\\
\mtr$_\beta$ \Rmnum{3}: $r(v)\in \mathbb{R}$.\\
\end{definition}

\subsection{Main Theorem}
We show that all versions of this problem are NP-hard:
\begin{theorem}
\mtr$_\beta$ \Rmnum{1},\Rmnum{2},\Rmnum{3} are NP-hard for all $\beta<1$.
\label{thm:MTR}
\end{theorem}

We prove the theorem by reducing 3-SAT to \mtr$_\beta$. The reduction graph is depicted in Figure \ref{graph3:1} (detailed description will follow). The proof idea is as below: We construct a long graph body with two long strings ($u_nv_nu_{n-1}v_{n-1}\cdots u_0$ and $u'_nv'_nu'_{n-1}v'_{n-1}\cdots u'_0$), and set weight in the graph such that there is reward on one and only one of $v_i$ and $v_i'$ in the minimum reward setting; also, the one with reward depends on the 3-SAT clause so as to construct a valid assignment.

\subsection{Graph construction}
For any 3-SAT instance, suppose it has $m$ clauses and $n$ variables. \\
Let $x=\beta^{-1}-1, y=\frac12 x$, and
\begin{eqnarray}
g_k&=&6(2n-k),\nonumber\\
h_k&=&6(2n-k)+6+y,\nonumber\\
r_t&=&12n-6+6\beta^{-1},\nonumber\\
l&=&\left\lceil \frac{nx+12n+6\beta^{-1}-6}{\beta x}\right\rceil+1.\nonumber
\end{eqnarray}
We choose $n$ large enough such that
\[nx=n(\beta^{-1}-1)>2\beta^{-1}.\]
We construct graph $G$ in the following way: We construct $l+1$ nodes for each clause $i$, namely $a_{i,1},a_{i,2},\cdots,a_{i,l},b_i$; and 4 nodes $u_i,u_i',v_i,v_i'$ for each variable $i$. We also have node $c^*,s$ (source) and $t$ (terminal). We construct edges using the following configuration: (A triple below denotes the (start node, end node, weight) of an edge )
\begin{enumerate}
\item $(s,a_{1,1},0);(a_{i,1},b_i,0);(a_{i,j},a_{i,j+1},\beta x);(a_{i,l},a_{i+1,1},0);(1\leq i\leq m-1,1\leq j\leq l-1)$
\item $(a_{m,l},c_1,0),(c_i,c_{i+1},6),(c_n,u_n,6),(c_n,u_n',6);(1\leq i\leq n-1)$
\item $(u_i,v_i,x);(u_i',v_i',x);(1\leq i\leq n)$
\item $(v_i,u_{i-1},6);(v_i,u_{i-1}',6);(v_i',u_{i-1},6);(v_i',u_{i-1}',6);(1\leq i\leq n)$
\item $(u_0,t,0);(u_0',t,0).$
\end{enumerate}


For a clause $i$ that contains variable $k$, we construct $(b_i,u_{k-1},h_k),(b_i,u_{k-1}',h_k)$, If it contains a positive subclause of $k$, we construct $(a_{i,j},v_k,g_k),2\leq j\leq l$; otherwise if it contains ''not variable $k$'', we construct $(a_{i,j},v_k',g_k),2\leq j\leq l$. An illustration of the graph is depicted in Figure \ref{graph3:1}. For simplicity, direct values of $h_k$ and $g_k$ are drawn. The graph shows an example of a clause with $v_n$ and $\neg v_k$.
\begin{figure}[htbp]
\centering\includegraphics[width=400px]{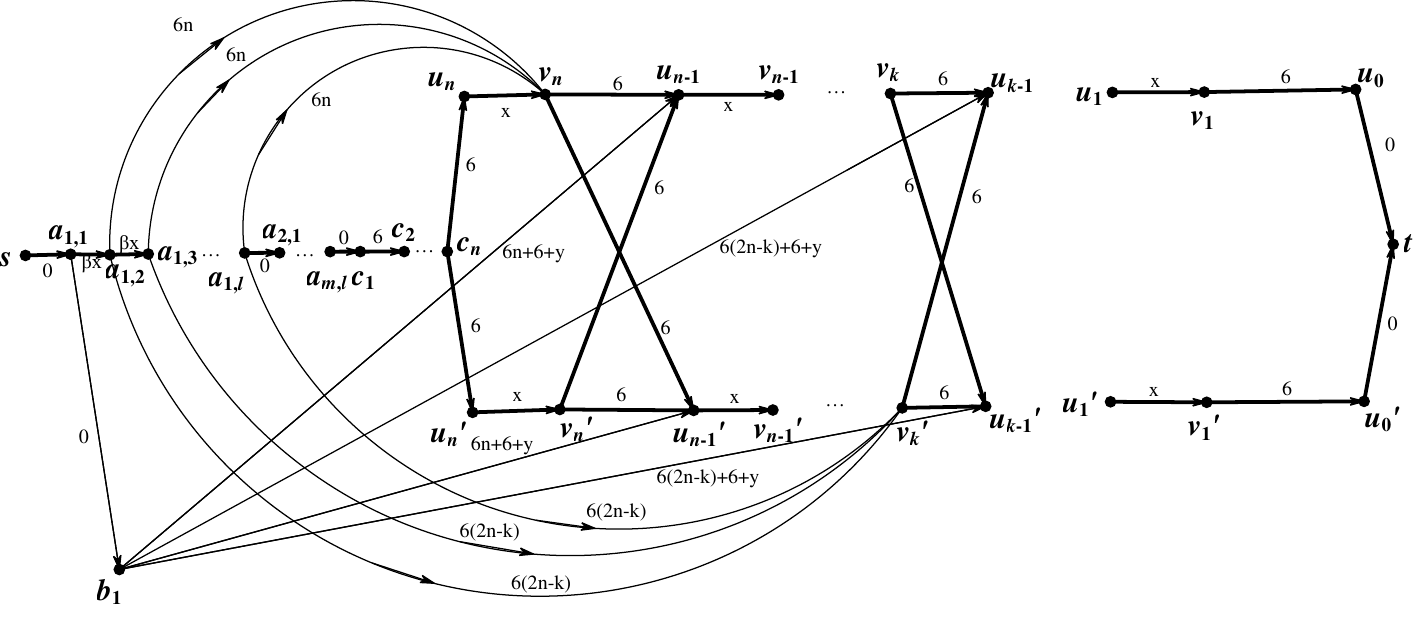}
\caption{Example figure for proving theorem \ref{thm:MTR}, with a clause with $v_n$ and $\neg v_k$. Edges that might be traveled by the agent are in bold lines. }\label{graph3:1}
\end{figure}
The feasibility of this graph is easily seen from the fact that $\beta<1$ (so $x>0$). We now prove the following main proposition:\\
\begin{proposition}
The 3-CNF is satisfiable if and only if $G$ has a minimum feasible configuration with sum of rewards at most $nx+r_t$.
\end{proposition}
\subsection{Two Lemmas}
To start with the proof, we first give two lemmas. The first lemma is about constraints in our graph:
\begin{lemma}
 We have
\begin{eqnarray}
h_k\beta^{-1}+6(k-1)&>&nx+r_t,\label{eqn1}\\
g_k\beta^{-1}+6k&>&nx+r_t,\label{eqn2}\\
(l-1)\beta x&>&r_t+nx,\label{eqn3}\\
h_k&<&x+g_k+6,\label{eqn4}\\
x+g_k-x+6k&<&r_t,\label{eqn5}\\
x+g_k-x+6k&<&h_k+6(k-1),\label{eqn6}\\
\beta x+g_k-x+6k&<&r_t,\label{eqn7}\\
6\beta^{-1}+6(n-1)+6n&=&r_t,\label{eqn8}\\
x+g_k-x+6k&<&g_k\beta^{-1}-x+6k.\label{eqn9}
\end{eqnarray}
\label{lem:eqns}
\end{lemma}
Each of the inequalities (equality) are necessary in guiding the agent to $t$. For example, (\ref{eqn1}) makes sure that the edge of $(b_i,u_{k-1})$ will not be used by the agent if the sum of rewards is at most $nx+r_t$. Details will be explained later in the main proof.
\begin{proof}
We have $h_k\beta^{-1}+6(k-1)=\beta^{-1}(6(2n-k)+6+y)+6(k-1)>\beta^{-1}6(2n-k)+6k=g_k\beta^{-1}+6k$.
On the other hand,
\begin{eqnarray*}
g_k\beta^{-1}+6k&=&\beta^{-1}6(2n-k)+6k\\
&\geq& 6n+6n\beta^{-1}\\
&=&12n+6nx\\
&>&nx+12n+6\beta^{-1}-6=nx+r_t.
\end{eqnarray*}
The last inequality is because the way we choose $n$: $nx>2\beta^{-1}$.
So we prove (\ref{eqn1}) and (\ref{eqn2}). \\
To prove (\ref{eqn3}), we choose an $l$ large enough: thus we have
\[r_t+nx=12n-6+6\beta^{-1}+nx\leq \beta x(l-1).\]

(\ref{eqn4}) follows directly from definition.

For (\ref{eqn5}) the left hand side is \(g_k+6k=6(2n-k)+6k=12n<r_t.\)

For (\ref{eqn6}) the right hand side is \(h_k+6(k-1)=12n+y>12n\) which is the left hand side.

(\ref{eqn7}) is a direct induction of (\ref{eqn5}), and (\ref{eqn8}) is the definition of $r_t$.

For (\ref{eqn9}), it is sufficient to prove $x<g_k(\beta^{-1}-1)$. In fact we have $g_k(\beta^{-1}-1)=g_kx=6(2n-k)x>x$.

So we prove all the inequalities and equalities.
\end{proof}
\indent \indent The second lemma is a simple equivalent condition for the agent to make next move; we use this condition to simplify our proof.
\begin{lemma}
The agent continues to move if and only if there exists a path $P$ such that
\[\beta^{-1}c(P)=\beta^{-1}c(e_1)+\sum_{e\in P:e\neq e_1} c(e) \leq \sum_{v\in P} r(v),\]
where $e_1$ is the first edge of $P$.
\label{lemmaeasy}
\end{lemma}
\begin{proof}
This lemma can be obtained simply by multiplying $\beta^{-1}$ in the definition formula at the beginning of this section.
\end{proof}

 By Lemma \ref{lemmaeasy}, the agent evaluates
\[\beta^{-1}c(e_1)+\sum_{e\in P:e\neq e_1} c(e)\]
instead of $c(P)$ for every path $P$ at each node, and choose the edge which minimizes this value.

\subsection{Main proof of theorem \ref{thm:MTR}}
\begin{proof}
We first prove the achievability of the bound $nx+r_t$ when the clause is satisfiable.
\begin{proposition}
If the 3-CNF is satisfiable, then $G$ has a minimum feasible configuration with sum of rewards at most $nx+r_t$.
\label{prop:ach}
\end{proposition}
\begin{proof}
If the 3-CNF is satisfiable, we construct a reward configuration as follows: Pick one valid assignment and let $r(v_i)=x$ if variable $i$ is true in the assignment, otherwise let $r(v_i')=x$. Also let $r(t)=r_t$. There's no reward on all the other nodes. We now prove that this is a valid setting: Define the length value of a path $P$ from $v$ to $v'$, denoted $l(P)$, to be the sum of all weights on the path minus sum of all rewards on the path except $v$ and $v'$:
\[l(P)=\sum_{e\in P}c(e)-\sum_{u\in P:u\ne v,v'} r(u).\]
The distance of a node $v$ to $t$, denoted $d(v)$, is defined as the minimum length value of all paths getting from $v$ to $t$. In fact, it's easy to see that in this setting all $v_i,v_i'$ have $d(v_i)=d(v_i')=6+(k-1)(6+x-x)=6k$, and if variable $i$ is true then $d(u_i)=6k$, otherwise $d(u_i')=6k$. By easy computation using lemma \ref{lem:eqns}, one can verify that the agent makes all the way to $t$ under the construction, thus finishing the proof; the details of this verification is provided in the appendix.
\end{proof}

We then prove that the satisfiability of the clause when $nx+r_t$ is achieved.
\begin{proposition}
If $G$ has a minimum feasible configuration with sum of rewards at most $nx+r_t$, then the 3-CNF is satisfiable,
\end{proposition}
\begin{proof}
First notice that in the minimum reward configuration the agent cannot traverse through any $b_i$ to $u_k$, since otherwise it requires at least
$\beta^{-1}h_k+6(k-1)>r_t+nx$ (by (\ref{eqn1})) reward.

On the other hand, in a minimum setting the edge $(a_{i,j},v_k)$ also cannot be used, since otherwise it requires at least $\beta^{-1}g_k+6k>r_t+nx$ (by (\ref{eqn2})) reward.

Combine the two facts above, we know that the agent will reach $c_1$. Notice that all path from $c_1$ to $t$ has the form $c_1c_2\cdots c_n u_n^*v_n^*u_{n-1}^*\cdots v_1^*t$, where $u_i^*\in\{u_i,u_i'\},v_i^*\in\{v_i,v_i'\}$. On $c_1$, the agent evaluates the cost of any of these paths (excluding the rewards) to be $6\beta^{-1}+6(n-1)+6n=r_t+nx$ (by (\ref{eqn8})), so at least $r_t+nx$ rewards has to be put on nodes $\{c_1,c_2,...,c_n,u_n,\cdots,u_1,u_n',\cdots,u_1',v_n,\cdots,v_1,v_n',\cdots,v_1',t\}$, in all three versions of \mtr$_\beta$. And so if $G$ has a minimum feasible configuration with sum of rewards at most $nx+r_t$, no reward is put on nodes $a_{i,j}$ or $b_i$. Also in \mtr$_\beta$ \Rmnum{3}, no negative rewards can be given. consider the path that the agent chooses the shortest on $c_1$, we know that there must be a path from $c_1$ to $t$ with a total reward (including reward on $t$) of $r_t+nx$. Then since there's no path from $v_k$ to $v_k'$, $u_k$ to $u_k'$ (vice versa), we can't put rewards on both $v_k$ and $v_k'$ or $u_k$ and $u_k'$ (the agent can't get both of them).

Consider the following assignment: if there's reward on node $v_k$, let variable $v_k$ = true; otherwise let variable $v_k$ = false. In the following, we prove this assignment gives a valid assignment of the original 3-SAT clause.

Note that if there's reward on $v_k$, then $v_k'$ doesn't have; vice versa. So to prove this proposition we only need to prove that for every clause $i$, among the three $v$ kind nodes (call them $v_{k_1}^*,v_{k_2}^*,v_{k_3}^*$, $v_{k_i}^*\in\{v_{k_i},v_{k_i}'\}$) that $a_{i,2}$ connects to, there's at least one with reward. To prove this, first notice that the shortest path from $a_{i,2}$ to $t$ must be through some $v_{k_j}^*$: The path through $a_{i,j}(j>2)$ and then to $v_{k_j}^*$ has same costs but an additional $\beta x$ cost; and the path from $a_{i+1,1}$ has cost at least $\beta x(l-1)>nx+r_t$ (by (\ref{eqn3}) and that there's no reward on $a_{i,j}$), which is larger than the total sum of rewards.\\

If there's no reward on $v_{k_1}^*,v_{k_2}^*$ or $v_{k_3}^*$, we have
\[d(a_{i,2})=\min_{1\leq j\leq 3, v_{k_j}^*} g_{k_j}+d(v_{k_j}^*)=\min_{1\leq j\leq 3} \min\{g_{k_j}+6+d(u_{k_j-1}),g_{k_j}+6+d(u'_{k_j-1})\}:=S.\]
So at $a_{i,1}$ the edge to $a_{i,2}$ is evaluated as $\beta x\cdot \beta^{-1}+S=x+S$, since there's no reward on $a_{i,1}$. On the other hand, the path to $b_i$ is evaluated as (note that there's no reward on $b_i$)
\begin{eqnarray*}
0+d({b_i})&=&\min_{1\leq j\leq 3} \min\{h_{k_j}+d(u_{k_j-1}),h_{k_j}+d(u'_{k_j-1})\}\\
&<&\min_{1\leq j\leq 3} \min\{g_{k_j}+x+6+d(u_{k_j-1}),g_{k_j}+x+6+d(u'_{k_j-1})\}(\text{by} (\ref{eqn4}))\\
&=&x+S.
\end{eqnarray*}
So the agent will go to $b_i$, which makes contradiction with above. So there's reward on at least one of $v_{k_1},v_{k_2},v_{k_3}$. Since it's impossible that we put reward on both $v_k$ and $v_k'$, we prove the proposition.
\end{proof}
\indent \indent Combine the two sides above, we prove Theorem \ref{thm:MTR}.
\end{proof}

\section{Conclusion and Further work}
In this paper, we solve three open problems raised by \cite{kleinberg2014time}. Our results show the close relationship between the special structure $\mathcal{F}_k$ and exponential cost ratio, and thus we should avoid such structure when designing structure of tasks. We also prove the hardness of motivating agents - either by deleting edges, or putting rewards. These problems seems to be intractable and we need a careful search when trying to motivate agents.


As an extension of our work, we can consider other forms of the problem, for example, for \mtr$_\beta$, consider the following extension:

\mtr$_\beta$ \Rmnum{4}: The designer can put any reward on any node, but only pays what the agent actually claims. As a result, we only want to minimize the claimed total rewards.

Problem \mtr$_\beta$ \Rmnum{4} seems to be NP-complete, but we are not able to show it using our method.

We are also interested in the case where multi-agents collaboratively reach the target node. Similar question asked in this paper can be asked in this extension as well.

\section{Acknowledgement}
This work was supported in part by the National Basic Research Program of China Grant 2011CBA00300, 2011CBA00301, the National Natural Science Foundation of China Grant 61033001, 61361136003, 61303077, and a Tsinghua University Initiative Scientific Research Grant.

\bibliographystyle{acmsmall}
\bibliography{acmsmall-sample-bibfile}


\newpage
\appendix
\section{Proof to lemmas in section 3}
In the following proofs, we only require the property of $t_i$ that $t_i\geq i+1$, which has been derived in the main proof. Also note that $S_i=\left\{j|1\leq j<i,t_j\geq i\right\}$.

\subsection{Proof of Lemma \ref{xsk:mathematical lemma}}
In addition, we prove
\begin{eqnarray*}
x_{m-1}+\sum_{j\in S_{m-1}}x_j=\sum_{j:t_j=m-1}x_j+\sum_{j:m-1<t_j<m}x_j+\sum_{j\in S_m}x_j,
\end{eqnarray*}
which will be used later.
\begin{proof}
For $m=2,3,\ldots,n$, note that $t_{m-1}\geq m$, we have $m-1\in S_m$ and $m-1\notin S_{m-1}$. Therefore,
\begin{eqnarray*}
S_{m-1}\cup \{m-1\}&=&\{m-1\}\cup\{j|1\leq j<m-1,t_j\geq m-1\}\\
&=&\{m-1\}\cup\{j|1\leq j<m-1,t_j\geq m\}{\cup}\: \{j|m-1\leq t_j<m\}\\
&=&S_m\cup \{j|t_j= m-1\}{\cup}\:\{j|m-1<t_j<m\},
\end{eqnarray*}
where sets on each side of the equality are disjoint. Thus,
\begin{eqnarray*}
x_{m-1}+\sum_{j\in S_{m-1}}x_j=\sum_{j:t_j=m-1}x_j+\sum_{j:m-1<t_j<m}x_j+\sum_{j\in S_m}x_j.
\end{eqnarray*}
Also note that
\begin{eqnarray*}
\sum_{j:t_j=m-1} x_j+\sum_{j:m-1<t_j<m}x_j=\sum_{j:m-1\leq t_j<m}x_j=\sum_{j\in S_{m-1}:m-1\leq t_j<m}x_j,
\end{eqnarray*}
we have
\begin{eqnarray*}
\sum_{j\in S_m} x_j=x_{m-1}+\sum_{j\in S_{m-1}} x_j-\sum_{j\in S_{m-1}:m-1\leq t_j<m}x_j=x_{m-1}+\sum_{j\in S_{m-1}:t_j\geq m} x_j.
\end{eqnarray*}
\end{proof}

\subsection{Proof of Lemma \ref{xsk:bound lemma}}
\begin{proof}
For convenience, we restate (\ref{xsk:relation between a and b}) here again:
\begin{eqnarray*}
a_1=1,a_i=\beta b_{i-1}+\sum_{j:i-1<t_j<i}(1-\beta)b_j,b_i=a_i+\sum_{j:t_j=i} (1-\beta)b_j.
\end{eqnarray*}
Note $b_i\geq a_i$ can be directly derived from the recursion formula, we only need to prove $a_i\geq \beta^{|S_i|}$. Now we suppose that $i$ is fixed.

We prove firstly that for $m=1,2,\ldots,n$,
\[a_m+\sum_{j\in S_m} (1-\beta)b_j=1.\]
This can be proved by induction on $m$: the base case of $m=1$ holds trivially.

For inductive case: if it holds for $m-1$,
\begin{eqnarray}
&&a_m+\sum_{j\in S_m} (1-\beta)b_j \nonumber\\
&=&b_{m-1}-(1-\beta)b_{m-1}+\sum_{j:m-1<t_j<m}(1-\beta)b_j+\sum_{j\in S_m} (1-\beta)b_j \nonumber\\
&=&a_{m-1}+\sum_{j:t_j=m-1} (1-\beta)b_j-(1-\beta)b_{m-1}+\sum_{j:m-1<t_j<m}(1-\beta)b_j+\sum_{j\in S_m} (1-\beta)b_j \nonumber\\
&=&a_{m-1}+\sum_{j\in S_{m-1}} (1-\beta)b_j \label{xsk:use of Lemma 1 in Lemma 2}\\
&=&1, \nonumber
\end{eqnarray}
where equality (\ref{xsk:use of Lemma 1 in Lemma 2}) follows from Lemma \ref{xsk:mathematical lemma}.

Let $S_i=\left\{j_1,j_2,\ldots, j_{|S_i|}\right\}$ where $j_1<j_2<\cdots<j_{|S_i|}$. Then for $l=1,2,\ldots,|S_i|$ we have $t_{j_l}\geq i>j_l$, thus for all $1\leq p\leq l\leq |S_i|$, we have $t_{j_p}\geq i\geq j_l+1>j_p$, thus $j_p\in S_{j_l+1}$. Now we can obtain for $l=1,2,\ldots,|S_i|$,
\begin{eqnarray*}
b_{j_l}+\sum_{p=1}^{l-1}(1-\beta)b_{j_p}=\beta b_{j_l}+\sum_{p=1}^{l}(1-\beta)b_{j_p}\leq a_{j_l+1}+\sum_{j\in S_{j_l+1}}(1-\beta)b_j=1.
\end{eqnarray*}
Note $a_i+\sum_{p=1}^{|S_i|}(1-\beta)b_{j_p}=1$, to prove $a_i\geq \beta^{|S_i|}$, we only need to prove $\sum_{p=1}^{|S_i|}(1-\beta)b_{j_p}\leq 1-\beta^{|S_i|}$, or more strongly $\sum_{p=1}^{l}(1-\beta)b_{j_p}\leq 1-\beta^{l}$ for $l=1,2,\ldots,|S_i|$.

We prove this by induction on $l$: again, the base case of $l=1$ holds trivially.

For inductive case: if it holds for $l-1$,
\begin{eqnarray*}
\sum_{p=1}^{l}(1-\beta)b_{j_p}&=&\sum_{p=1}^{l-1}(1-\beta)b_{j_p}+(1-\beta)b_{j_l}\\
&\leq&\sum_{p=1}^{l-1}(1-\beta)b_{j_p}+(1-\beta)\left(1-\sum_{p=1}^{l-1}(1-\beta)b_{j_p}\right)\\
&=&\beta\sum_{p=1}^{l-1}(1-\beta)b_{j_p}+1-\beta\\
&\leq&\beta\left(1-\beta^{l-1}\right)+1-\beta\\
&=&1-\beta^l.
\end{eqnarray*}
\end{proof}

\subsection{Proof of Lemma \ref{xsk:graph lemma}}
\begin{proof}

\begin{figure}[h]
\centerline{\includegraphics[scale=0.6]{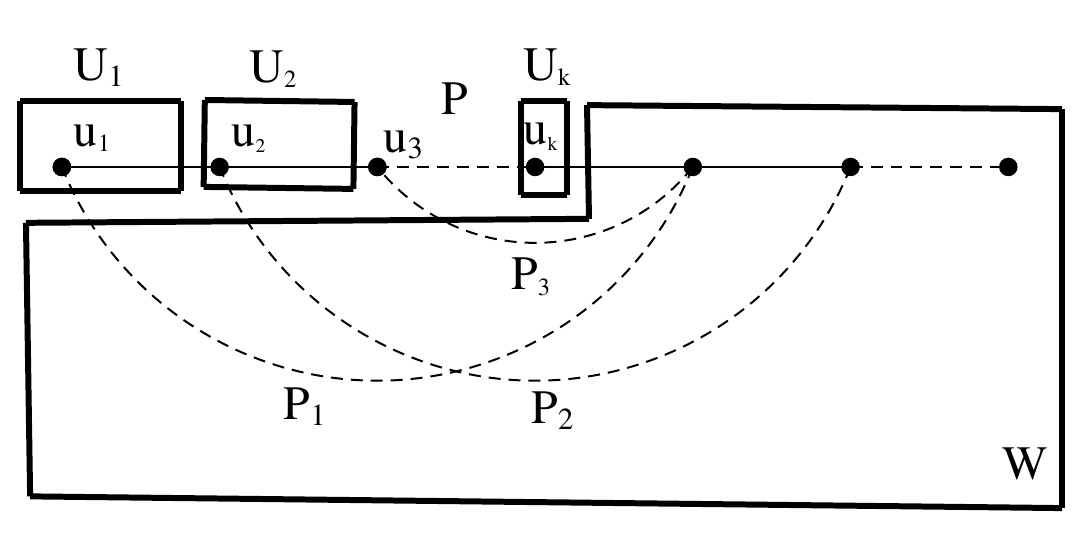}}
\caption{the graph contains an $\mathcal{F}_k$-minor.}
\label{fkminor}
\end{figure}
Let $U_1,U_2,\ldots,U_k,W$ be subgraphs of $\sigma(G)$ such that there is an edge between two nodes in one subgraph if and only if this edge exists in $\sigma(G)$. The nodes of these subgraphs are defined as follows: for $i=1,2,\ldots,k-1$, nodes of $U_i$ are nodes between $u_i$ and $u_{i+1}$ in $P$ (including $u_i$ and excluding $u_{i+1}$); $U_k$ contains only $u_k$; and nodes of $W$ are the union of $k$ parts $V_1,V_2,\ldots,V_{k-1},W'$ where $V_i$ are nodes between $u_i$ and $u_i'$ in $P_i$ (including $u_i'$ and excluding $u_i$) for $i=1,2,\ldots,k-1$ and $W'$ are all nodes after $u_k$ in $P$ (excluding $u_k$). See Figure \ref{fkminor}.

Clearly, $U_1,U_2,\ldots,U_k,W$ are disjoint while each subgraph is connected. For $i=1,2,\ldots,k-1$, $U_i$ and $U_{i+1}$ are connected by the edge immediately before $u_{i+1}$ in $P$, and $U_i$ and $W$ are connected by the edge immediately after $u_i$ in $P_i$. $U_k$ and $W$ are connected by the edge immediately after $u_k$ in $P$. Thus $\{U_1,U_2,\ldots,U_k,W\}$ can be considered as an $\mathcal{F}_k$-minor, which finishes the proof.
\end{proof}

\section{Verification in proposition \ref{prop:ach}}
We say a node $p$ will \emph{lead to} another node $q$, if at node $p$ the agent will choose the path from $p$ to $q$. We thus notice the following statements:
\begin{enumerate}
\item
node $s$ will lead to $a_{1,1}$ because (\ref{eqn7}) ($
d(a_{1,1})\leq \beta x+g_k-x+d(v_k)=x+g_k-x+6k\leq r_t
$).

\item $a_{i,1}$ will lead to $a_{i,2}$: The path that goes through $a_{i,2}$ and then $v_k$ (not losing generality, suppose $v_k$ is true) has length value $x+g_k-x+d(v_k)=x+g_k-x+6k$, and the path that goes through $b_i$ has length value $h_k+6(k-1)$. So we derive this from (\ref{eqn5}) and (\ref{eqn6}).

\item $a_{i,j}, j<l$ will lead to $a_{i,j+1}$, because (\ref{eqn5}) and (\ref{eqn9}) (The path that uses $a_{i,j}$ to $v_k$ has length value $g_k\beta^{-1}+6k$).

\item $a_{i,l},i<m$ will lead to $a_{i+1,1}$, because the path from $a_{i+1,1}\rightarrow a_{i+1,2}\rightarrow v_k\rightarrow t$ has length value $\beta x + g_k -x +6k$. So we derive this from (\ref{eqn7}) and (\ref{eqn9}).

\item $a_{m,l}$ will lead to $c_1$, because (\ref{eqn8}).

\item $c_i$ will lead to $c_{i+1}$ for $1\leq i\leq n-1$, because (\ref{eqn8}).

\item $c_n$ will lead to $u_n^*$ (here $u_n^*=u_n$ if variable $n$ is true, otherwise $u_n^*=u_n'$) because (\ref{eqn8}).

\item The agent then travels from $u_n^*$ to $t$ , which is trivial by computation on the weights.
\end{enumerate} 
\end{document}